\documentclass[11pt]{article}

\usepackage{amsmath,amssymb,amsfonts,amsthm} 
\usepackage{mathrsfs,latexsym} 
\usepackage{mathtools}

\usepackage{calrsfs}
\DeclareMathAlphabet{\pazocal}{OMS}{zplm}{m}{n}

\usepackage[mathscr]{eucal} 

\usepackage{color}
\usepackage{cite}
\usepackage{epsfig}
\usepackage{graphicx}

\usepackage{bm}
\usepackage{cite}

\newtheorem{theorem}{Theorem}[section]

\newtheorem{lemma}[theorem]{Lemma}

\numberwithin{equation}{section}

\newcommand{\II}{{\mathbb I}}

\newcommand{\RR}{{\mathbb R}}

\newcommand{\Xs}{{\mathscr{X}}}

\newcommand{\bgamma}{{\bm \gamma}}
\newcommand{\bsigma}{{\bm \sigma}}
\newcommand{\btau}{{\bm \tau}}

\newcommand{\bai}{{\bm a}}
\newcommand{\bbi}{{\bm b}}
\newcommand{\bfi}{{\bm f}}
\newcommand{\bgi}{{\bm g}}
\newcommand{\bni}{{\bm n}}

\newcommand{\bxi}{{\bm x}}
\newcommand{\byi}{{\bm y}}
\newcommand{\bAi}{{\bm A}}
\newcommand{\bBi}{{\bm B}}

\newcommand{\bEi}{{\bm E}}
\newcommand{\bOi}{{\bm O}}

\newcommand{\Af}{{\mathfrak A}}

\newcommand{\Wf}{{\mathfrak W}}

\newcommand{\gradient}{\text{grad}}
\newcommand{\divergence}{\text{div}}
\newcommand{\curl}{\text{curl}}

\newcommand{\Test}{\pazocal{D}}

\def\ie{{\it i.e.\ }}

\begin{document} 

%

\title{\LARGE On string-localized potentials and gauge fields}
\author{Detlev Buchholz${}^a$,
Fabio Ciolli${}^b$, Giuseppe Ruzzi${}^b$ and   
Ezio Vasselli${}^b$ \\[20pt]
\small  
${}^a$ Mathematisches Institut, Universit\"at G\"ottingen, \\
\small Bunsenstr.\ 3-5, 37073 G\"ottingen, Germany\\[5pt]
\small
${}^b$ 
Dipartimento di Matematica, Universit\'a di Roma ``Tor Vergata'' \\
\small Via della Ricerca Scientifica 1, 00133 Roma, Italy \\
}
\date{}

\maketitle

{\small 
\noindent {Abstract. A recent idea,  
put forward by Mund, Rehren and Schroer in~\mbox{\cite{Mu, MuReSch, Sch}}, 
is discussed; it suggests that in gauge quantum field theory 
one can replace the point-localized gauge fields by  
string-localized vector potentials built from gauge invariant
observables and a principle of 
string-independence.  Based on a kinematical 
model, describing unmovable (static) fields carrying 
opposite charges, it is shown that these string-localized potentials cannot 
be used for the description of the gauge bridges between
electrically charged fields. 
These bridges are needed in order to ensure the validity of Gauss' law.  
This observation does not preclude the existence of 
Poincar\'e invariant theories, describing the coupling of string-localized 
gauge invariant potentials to matter fields. 
But these potentials are not a full-fledged substitute for the gauge fields in 
``usual'' quantum electrodynamics.}
} \\[2mm] 

\section{Introduction}
\label{sec1}

It is a well known fact that tensor fields in local quantum field 
theory, which are closed two-forms on Minkowski space, 
need not be exact, the most 
prominent example being the electromagnetic field. One can resolve 
this cohomological obstruction if one proceeds to indefinite metric 
extensions of the physical Hilbert space, where the tensor fields
can be represented as curls (exterior derivatives) of local 
vector-valued gauge fields. This fact plays 
an important role in the perturbative construction of gauge 
quantum field theories. However,  
according to folklore, the cohomological difficulties 
can also be resolved within the standard Hilbert space framework of 
quantum field theory by giving up the locality condition.
This can easily be seen by looking in momentum space at the equation, 
\mbox{expressing} the fact that a field $F$ is a closed 
two-form (\ie the homogeneous Maxwell equation),
\[
p_\gamma \, \widetilde{F}_{\alpha \beta}(p) - 
p_\beta \, \widetilde{F}_{\alpha \gamma}(p) + 
p_\alpha \, \widetilde{F}_{\beta \gamma}(p) = 0 \, .
\]
Upon contracting this equation with a fixed vector $e^\gamma$ and
dividing it (in the sense of distributions) by $e p$ one 
finds that $F$ can be represented as curl of some non-local vector 
potential on the original Hilbert space. This field is
represented in configuration space by a 
string-localized integral of~$F$. 

\medskip
In recent articles, Mund, Rehren and 
Schroer \cite{Mu, MuReSch, Sch} have emphasized
that the latter field has better ultraviolet properties; this   
should help to improve the perturbative construction 
of field theories, where string localized vector fields, arising from local
tensors, are coupled
to local matter fields. Since the direction of the chosen vector 
$e$ ought to be irrelevant for the resulting theory, the authors 
proposed to solve the renormalization problems arising in the 
perturbative construction of string localized fields by a  
condition of string-independence. Indeed, first results 
seem to indicate that this idea is meaningful \cite{MuReSch}. Yet 
it remained unclear whether the string-localized observable vector fields
can replace the local gauge fields 
appearing in the gauge bridges, connecting electrically charged fields. 
These gauge bridges are an important ingredient 
in the construction of the electrically charged superselection sectors. 

\medskip
It is the aim of the present letter to shed light on this
problem. For the sake of simplicity and concreteness, we consider 
the limiting case of quantum electrodynamics, where the 
gauge field is coupled to static (infinitely heavy,   
unmovable) fields carrying an electric charge. 
This assumption allows us to analyze the model at the kinematical 
(fixed time) level on $\RR^3$. A convenient gauge in this limit 
is the so-called $A_0 = 0$ (temporal) gauge. There the spatial components 
$\bAi$, $\bEi$ of the vector potential, respectively electric field,  
satisfy canonical commutation
relations. They can be represented on Hilbert spaces,  
where the remaining spatial gauge transformations act 
by unitary operators. The electric field 
$\bEi$ and magnetic field $\bBi \doteq \curl \, \bAi$
are gauge invariant and hence are observable.
Representations of the observable algebra, 
describing the electromagnetic field in the absence of
charged fields, are characterized by the 
condition that the charge density $\divergence \, \bEi$ vanishes.
In representations describing static 
charges, which are localized at given points in $\RR^3$, the localization 
of the observables has to be restricted to the complement 
of these singularities. Still, the charge density 
has to vanish in this complement.

\medskip
The second de Rham cohomology of the punctured space $\RR^3$ does 
not vanish, comprising the possible values of the static charges. 
In spite of the fact that the positions of the static charges are not
accessible, it ought to be possible according to Gauss' law 
to determine their values by surface integrals of the electric field, 
surrounding their positions; these surface integrals are called flux 
operators. In order to comply with this condition, opposite charges must 
be connected by gauge bridges, including the possibility that some 
bridges extend to spatial infinity if the total charge is different
from zero.  The gauge bridges are described by automorphisms of the 
algebra of observables and are constructed from the local vector
potential $\bAi$. Their action on flux operators which are localized
in the spatial complement of the charged fields sharply 
indicates the value of the charges 
sitting at the endpoints of the bridges. So the charge content
of states can be determined in this manner.

\medskip
On the other hand, as we shall see,  
all automorphisms induced by observables act trivially
on the flux operators, so they cannot be used to 
replace the local gauge fields in the gauge bridges. 
In particular, they do not 
encode information about the values of the static charges.
This holds also true for limits of these automorphisms, 
obtained for example by going to the limiting case 
of string-localized observable fields. 
Since none of these automorphisms induces Gauss' law,
which is a characteristic of the electric charge, we are led 
to the conclusion  that the string-localized observable 
vector fields considered in \cite{Mu, MuReSch, Sch} do not replace the
local gauge fields in every respect.

\medskip
Our article is organized as follows. In the subsequent section we 
specify the model, establish our notation and recall some basic 
cohomological facts.
Sec.~3 contains the definitions of gauge bridges and flux operators
and an analysis of their properties. In Sec.~4 the properties of
representations containing static charges are discussed and the
article closes with some brief conclusions.

\section{The kinematical algebra}
\label{sec2}

Since we want to study the quantized 
electromagnetic field in presence of static (unmovable) charges,
we will work at fixed time and consider 
the so-called $A_0 = 0$ (temporal) gauge \cite{LoMoSt}. 
There the spatial parts $\bAi$ of the vector potential 
and $\bEi$ of the electric field satisfy
canonical commutation relations in Weyl form. 
Denoting by $\Test(\RR^3)$ the space of real vector-valued test
functions with compact support, the unitary Weyl operators are
given by $e^{i\bAi(\bfi)}, e^{i\bEi(\bgi)}$, $\bfi,\bgi \in \Test(\RR^3)$;
derivatives of the fields are defined in the sense of distributions.
The Weyl operators satisfy the relations 
\begin{equation}  \label{e2.1} 
e^{i\bAi(\bfi)} e^{i\bEi(\bgi)} 
= e^{i \langle \bfi, \bgi \rangle } 
e^{i\bEi(\bgi)} e^{i\bAi(\bfi)} \, ,
\end{equation}
where $\langle \bfi, \bgi \rangle = \int \! d\bxi \, \bfi(\bxi) \bgi(\bxi)$.
Weyl operators depending only on the vector potential $\bAi$, or on the electric
field $\bEi$, commute amongst each other and form unitary representations of
the vector space $\Test(\RR^3)$. The full set of
Weyl operators generates a C*-algebra~$\Wf(\RR^3)$, the
Weyl algebra. 
It contains for any given open region ${\bm O} \subset \RR^3$ a  
subalgebra $\Wf({\bm O})$ which is generated by all Weyl operators
assigned to test functions $\bfi, \bgi
\in \Test({\bm O})$, the subspace of functions having support
in ${\bm O}$. Operators localized in disjoint regions commute.

\medskip
On $\Wf(\RR^3)$ act the spatial gauge transformations 
$\bAi \mapsto \bAi + \gradient \, \bm s$, where~$\bm s$ is any
scalar test function. These transformations are induced by the adjoint
action of the Weyl operators
$e^{i \bEi(\text{grad} \, s)} = e^{- i (\divergence \, \bEi)(s)}$. 
The electric field $\bEi$
and the magnetic field $\bBi \doteq \curl \, \bAi$ are
gauge invariant. They generate
the algebra of observables $\Af(\RR^3) \subset \Wf(\RR^3)$.
This observable algebra contains for any open region 
${\bm O} \subset \RR^3$ a local subalgebra $\Af({\bm O})$, 
which is the gauge invariant part of $\Wf({\bm O})$.
The full algebra $\Af(\RR^3)$
can be used if observations are possible everywhere in~$\RR^3$. 
Yet if static charged fields
occupy a given set of spacetime points
$\Xs  \doteq \{ \bxi_1, \dots , \bxi_n \}$, these 
singular points cannot be accessed. One must then proceed  to the subalgebras  
$\Af(\RR^3 \backslash \Xs )$.

\medskip
Before turning to the analysis of the model, we compile
some notation. We will consider integrals of field operators 
along smooth, simple paths and across closed surfaces, which 
by definition are the 
boundary of some region in  $\RR^3$. In order to
mollify these integrals, we 
pick any scalar test function $\bxi \mapsto s(\bxi)$,
which is normalized according to  $\int \! d\bxi \, s(\bxi) = 1$,  
and make use of the following definitions of specific test functions. 
\begin{itemize}
\item[(i)] For given path $\bgamma : [0,1] \rightarrow \RR^3 $, 
  we put 
  \[ \bxi \mapsto
  \bgi_\bgamma(\bxi) \doteq \int_0^1 \! du \, s(\bxi - \bgamma(u)) \, 
  \dot{\bgamma}(u) \, \in \, \Test(\RR^3) \, , \] 
  where $\dot{\bgamma}$ denotes the tangent vector of $\bgamma$.
\item[(ii)] For given surface
  $\bsigma : [0,1]^2 \rightarrow \RR^3 $, we put  
  \[ \bxi \mapsto \bgi_\bsigma(\bxi) \doteq \iint_{[0,1]^2} d^2u \,
  s(\bxi - \bsigma(u)) \, \bni_\bsigma(u) \, \in \, \Test(\RR^3) \, , 
  \]
  where $\bni_\bsigma$ is the surface normal and $d^2u \, |\bni_\bsigma(u)|$ 
  the surface form on $\bsigma$. 
\item[(iii)] For given characteristic function
  $\btau : [0,1]^3 \rightarrow \RR^3 $, fixing a region, we put
  \[ \bxi \mapsto s_\btau(\bxi) \doteq \iiint_{[0,1]^3} d^3u \,
  s(\bxi - \btau(u)) \, , 
  \] 
  which is a scalar test function. 
\end{itemize}
Choosing functions $s$ with sufficiently small support 
about the origin, the resulting test functions
have supports in arbitrarily small neighborhoods of the
given regions. For fixed $s$, the following 
consequences of Gauss' and Stokes' laws are well-known, where $\partial$ 
denotes the boundary of the respective regions. 

\begin{itemize}
\item[(i)] $\gradient \, s_\btau = - \bgi_{\bsigma}$, where 
$\bsigma = \partial \btau$
\item[(ii)] $\curl \, \bgi_{\bsigma} = - 
\bgi_\bgamma$, where $\bgamma = \partial \bsigma$
\item[(iii)] $\divergence \, \bgi_\bgamma = 
s(\bgamma(0)) - s(\bgamma(1))$.   
\end{itemize}
We also note that 
$\langle \curl \, \bfi, \bgi \rangle = \langle \bfi, \curl \, \bgi \rangle
$ for $\bfi, \bgi \in \Test(\RR^3)$.

\section{Flux operators and gauge bridges}
\label{sec3}

Let $\Xs  \doteq \{ \bxi_1, \dots , \bxi_n \}$ be any given 
(possibly empty) set of points which are removed from $\RR^3$.
We analyze in this section the properties of Weyl operators which 
describe flux integrals in $\Af(\RR^3 \backslash \Xs)$,
respectively gauge bridges between points in $\Xs$.

\medskip
We begin with the flux integrals. Given any 
closed surface $\bsigma \subset \RR^3 \backslash \Xs$, we 
pick a corresponding test function $\bgi_\bsigma$ which also has
support in $\RR^3 \backslash \Xs$. The resulting Weyl operator 
$e^{i \bEi(\bgi_\bsigma)}$ is an element of $\Af(\RR^3 \backslash \Xs)$.
It commutes with all Weyl operators depending only on 
the electric field $\bEi$. According to the Weyl relations, the
commutators with the magnetic field $\bBi = \curl \, \bAi$ 
are given by
\[
e^{i \bBi(\bfi)} e^{i \bEi(\bgi_\bsigma)} = 
e^{i \langle \curl \, \bfi, \, \bgi_\bsigma \rangle} \, 
e^{i \bEi(\bgi_\bsigma)} e^{i \bBi(\bfi)} \, , \quad \bfi \in 
\Test(\RR^3 \backslash \Xs) \, .
\]
Now 
\[ \langle \curl \, \bfi, \, \bgi_\bsigma \rangle 
= \langle \bfi, \, \curl \, \bgi_\bsigma \rangle 
= \langle \bfi, \, \bgi_{\, \partial \bsigma} \rangle = 0
\]
since $\bsigma$ has no boundary, $\partial \bsigma = \emptyset$.
Thus $e^{i \bEi(\bgi_\bsigma)}$ commutes also with all magnetic fields. This 
leads us to our first result.

\begin{lemma} \label{l3.1}
Let $\bsigma \subset \RR^3 \backslash \Xs$ be a closed surface and let 
$\bgi_\bsigma$ be a corresponding test function, which also has support in 
$\RR^3 \backslash \Xs$. 
Then $e^{i \bEi(\bgi_\bsigma)}$ is contained in the center of
$\Af( \RR^3 \backslash \Xs)$. These elements are called flux operators. 
\end{lemma}

\medskip
Next, we turn to the gauge bridges. These bridges are defined by
the adjoint action of Weyl operators, involving the vector 
potential $\bAi$, on the observables. As a matter of fact, since
these bridges connect static charges, we have to proceed to 
limits of Weyl operators which are no longer elements of 
the Weyl algebra. But, as we shall see, their adjoint action 
on the observables is still
well defined. 

\medskip
Let $\bxi_1, \bxi_2 \in \Xs$ be a pair of 
points which is occupied by static charges. For the construction
of gauge bridges between these points we need to exhibit specific
test functions. To this end we fix a simple path 
$\bgamma_0 : [0,1] \rightarrow \RR^3$ which connects the two points,
\ie $\bgamma_0(0) = \bxi_1, \gamma_0(1) = \bxi_2$, does not meet 
other points in $\Xs$, and satisfies  
\[
\inf_{u \in (0,1)} \, u^{-1} (1 - u)^{-1}
| \bgamma_0(u) - \bxi_{1/2} | \geq R > 0 \, . 
\]
The neighboring paths $\bgamma_\byi : [0,1] \rightarrow \RR^3$, 
$\byi \in \RR^3$, given by
\[
u \mapsto \bgamma_\byi(u) \doteq \bgamma_0(u) + u(1-u) \, \byi \, , \quad 
 |\byi| < R \, ,
\]
also connect the two given points and, reducing the
value of $R$ if necessary, do not meet any other points in $\Xs$. 
Next, we pick some scalar test function $s$ which has 
support in a ball around the origin of radius $R$ and 
satisfies $\int \! d\bxi \, s(\bxi) = q$ for given 
$q \in \RR$. We then consider
for $u \in (0,1)$ the functions 
\[
\bxi \mapsto \bai_{u}(\bxi) \doteq \int \! d\byi \, s(\byi) \, 
\delta(\bxi - \bgamma_\byi(u)) \, \dot{\bgamma_\byi}(u) \, ,
\]
where $\delta$ denotes the Dirac measure.
For the chosen paths, these integrals can easily be 
computed and it is then apparent that the
functions $\bai_u$ are test functions having support 
in $\RR^3 \backslash \Xs$. Moreover their spatial 
derivatives on compact subsets 
of $\RR^3 \backslash \Xs$ are uniformly continuous in $u$ and 
\[
\bxi \mapsto \divergence \, \bai_u(\bxi)
= - \mbox{\Large $\frac{d}{du}$} \int \! d\byi \, s(\byi) \, 
\delta(\bxi - \bgamma_\byi(u)) \, .
\]
We can proceed now to test functions 
used for the construction of gauge bridges. Given
$\varepsilon > 0$, we put
\[
\bxi \mapsto \bbi_\varepsilon(\bxi) \doteq 
\int_\varepsilon^{1-\varepsilon} \! du \, \bai_u(\bxi) \in \Test(\RR^3) \, ,
\]
so the functions $\bbi_\varepsilon$ inherit the support and 
smoothness properties of their precursors $\bai_u$, and 
\[
\lim_{\varepsilon \searrow 0} \, \divergence \, \bbi_\varepsilon(\bxi) = 
q \, \delta(\bxi - \bxi_1) - q \, \delta(\bxi - \bxi_2) 
\] 
in the sense of distributions. 


\medskip 
After these preparations, we 
can define automorphisms $\beta_{q, \,  \bxi_1 \, \bxi_2}(\varepsilon)$ 
of the Weyl algebra, which approximate a gauge bridge
from $\bxi_1$ to $\bxi_2$ for given charge~$q$.
They act trivially on Weyl operators depending on $\bAi$ and
act on Weyl operators involving the electric field $\bEi$ according to   
\[
\beta_{q, \,  \bxi_1 \, \bxi_2}(\varepsilon)(e^{i \bEi(\bgi)}) \doteq 
 e^{i \bAi(\bbi_\varepsilon)}  e^{i \bEi(\bgi)} e^{-i \bAi(\bbi_\varepsilon)} 
= e^{i \langle \bbi_\varepsilon, \bgi \rangle} \,  e^{i \bEi(\bgi)} \, ,
\quad \bgi \in \Test({\RR^3}) \, .
\]
It follows from the preceding discussion that there is some distribution 
$\bbi_0$ such that for any $ \bgi \in \Test({\RR^3})$ 
and scalar test function $s^\prime$ on $\RR^3$ 
\[
\lim_{\varepsilon \searrow 0} \, \langle \bbi_\varepsilon, \bgi \rangle
= \langle \bbi_0, \bgi \rangle \quad \text{and} \quad
\lim_{\varepsilon \searrow 0} \, \langle \divergence \, \bbi_\varepsilon,
s^\prime \rangle = q s^\prime(\bxi_1) - q s^\prime (\bxi_2) \, . 
\] 
Since the polynomials of Weyl operators are norm dense in the 
Weyl algebra, it implies that the automorphisms 
$\beta_{q, \,  \bxi_1 \, \bxi_2}(\varepsilon)$, 
restricted to the observable algebra $\Af(\RR^3)$,
converge in this limit pointwise in the norm topology to some automorphism 
$\beta_{q, \,  \bxi_1 \, \bxi_2}$,
\begin{equation} \label{e3.1}
\lim_{\varepsilon \searrow 0} \beta_{q, \,  \bxi_1 \, \bxi_2}(\varepsilon)(A) = 
\beta_{q, \,  \bxi_1 \, \bxi_2}(A) \, , \quad 
A \in \Af(\RR^3) \, .
\end{equation}
This automorphism defines a gauge bridge. 
Since the underlying test functions $\bbi_\varepsilon$ have support in 
$\RR^3 \backslash \Xs$, it is apparent that the restriction of 
$\beta_{q, \,  \bxi_1 \, \bxi_2}$ 
to $\Af(\RR^3 \backslash \Xs)$ maps this subalgebra 
onto itself. Some relevant 
properties of this restriction are compiled in the following lemma.

\begin{lemma} \label{l3.2}
Let $\beta_{q, \,  \bxi_1 \, \bxi_2} \upharpoonright \Af(\RR^3 \backslash \Xs)$ 
be the restriction of the automorphism defined in equation
\eqref{e3.1}. Then
\begin{itemize} 
\item[(i)] for any scalar test function $s$, having support in 
$\RR^3 \backslash \Xs$, 
the operators $e^{i \, \divergence \, \bEi(s)}$ are invariant under the 
action of $\beta_{q, \,  \bxi_1 \, \bxi_2}$.
\item[(ii)] for any  
closed surface $\bsigma \subset {\RR^3 \backslash \Xs}$, 
the flux operators for the corresponding 
surface functions~$\bbi_\bsigma$, having support 
in a sufficiently small neighborhood of $\bsigma$,  satisfy 
\[
\beta_{q, \,  \bxi_1 \, \bxi_2}(e^{i \bEi(\bgi_\bsigma)}) = 
\begin{cases}
e^{i \bEi(\bgi_\bsigma)} & \text{if} \ \bsigma 
\ \text{encloses no point in} \ \Xs \\ 
e^{i q} \, e^{i \bEi(\bgi_\bsigma)} & \text{if} \ \bsigma \ 
\text{encloses only} \ \bxi_1 \\
e^{-iq} \, e^{i \bEi(\bgi_\bsigma)} & \text{if} \ \bsigma \ 
\text{encloses only} \ \bxi_2 \\
e^{i \bEi(\bgi_\bsigma)} & \text{if} \ \bsigma \ \text{encloses, both,} \
\bxi_1 \ \text{and} \ \bxi_2.
\end{cases}
\]
\item[(iii)] for any relatively compact,  
contractible region $\bOi \subset {\RR^3 \backslash \Xs}$
there exists some test 
function \mbox{$\bfi_{\, \bOi} \in \Test(\RR^3 \backslash \Xs)$} such that 
\[
\beta_{q, \,  \bxi_1 \, \bxi_2}(A) = 
e^{i \bBi(\bfi_{\, \bOi})} A e^{-i \bBi(\bfi_{\, \bOi})} \, , \quad A \in \Af(\bOi) \, .
\] 
\end{itemize}
\end{lemma}
\begin{proof}
(i) Recalling that 
$\, \divergence \, \bEi(s) = - \bEi(\gradient \, s) \,$, the statement
follows from 
\[
- \langle \bbi_0, \, \gradient \, s \rangle 
= - \lim_{\varepsilon \searrow 0} \, \langle \bbi_\varepsilon, \,
\gradient \, s \rangle  
= \lim_{\varepsilon \searrow 0}  \, \langle \divergence \, \bbi_\varepsilon, 
\, s \rangle  = q \, s(\bxi_2) - q \, s(\bxi_1) = 0 \, .
\] 
(ii) As was mentioned in the preceding section, one has 
$\, \bgi_\bsigma = - \gradient \, s_\btau$. Hence 
\[
\langle \bbi_0, \, \bgi_\bsigma \rangle =
- \lim_{\varepsilon \searrow 0} \ \langle \bbi_\varepsilon, \, \gradient \, s_\btau
 \rangle 
= \lim_{\varepsilon \searrow 0} \ \langle \divergence \, \bbi_\varepsilon, \, s_\btau 
\rangle 
= q \, s_\btau(\bxi_1) - q \, s_\btau(\bxi_2) \, . 
\] 
Since $\bxi \mapsto s_\btau(\bxi)$ are mollified characteristic functions
of the region $\btau$, the statement follows. 

\medskip
(iii) For any scalar test function $s$ with support in $\bOi$, one has
\[
\langle \bbi_0 , \gradient \, s \rangle =
\lim_{\varepsilon \searrow 0} \, 
\langle \bbi_\varepsilon, \gradient \, s \rangle 
= \lim_{\varepsilon \searrow 0} \,
- \langle \divergence \, \bbi_\varepsilon, 
\, s \rangle = 0 
\]
since the points $\bxi_1, \bxi_2$ lie in the complement of the
support of $s$. 
Since all derivatives of $\bbi_\varepsilon$ exist on $\bOi$ and are
uniformly continuous in $\varepsilon$, the restriction of the 
distribution $\bbi_{\, 0}$ to $\Test(\bOi)$ can be represented by a 
smooth function $\bbi_{\bOi}$ and 
$\langle \divergence \, \bbi_{\bOi}, s \rangle =
-  \langle \bbi_{\, 0}, \gradient \, s \rangle = 0
$. 
Hence $\divergence \, \bbi_{\bOi} = 0$ on $\bOi$, 
and since $\bOi$ is contractible,
there exists  by Poincar\'e's lemma a smooth function $\bfi_{\bOi}$ 
such that 
$\bbi_{\bOi} = \curl \, \bfi_{\bOi}$. 
Extending $\bfi_{\bOi}$ to a test function
having support in a sufficiently small neighbourhood of $\bOi$,  it follows
that $\langle  \bbi_{\, 0}, \bgi \rangle = 
\langle \curl \, \bfi_{\bOi}, \bgi \rangle $, $\bgi \in \Test(\bOi)$.
Hence
\[
\beta_{q, \,  \bxi_1 \, \bxi_2}(e^{i \bEi(\bgi)}) =
e^{i \langle \bbi_{\, 0}, \bgi \rangle} \, e^{i \bEi(\bgi)} =
e^{i \langle \curl \, f_{\bOi}, \bgi \rangle} \, e^{i \bEi(\bgi)} =
e^{i \bBi(\bfi_{\bOi})} e^{i \bEi(\bgi)}  e^{-i \bBi(\bfi_{\bOi})} \, ,
\]
completing  the proof. 
\end{proof}
\noindent Having clarified the properties of the flux operators and 
gauge bridges, we can turn now to the analysis of states. 

\section{States containing  static charges}
\label{sec4}

The observable Weyl algebra $\Af(\RR^3)$ 
has an abundance of regular irreducible representations, where
the generators of the Weyl operators are densely defined. We pick any
such state $\omega_{\, \emptyset}$ which does not contain charged 
fields. So the charge density vanishes in this state, \ie
$ \omega_{\, \emptyset}(e^{i \, \divergence \, \bEi(s)}) = 1$  
for all scalar test functions $s$ on $\RR^3$. An 
example satisfying this condition is the vacuum state of the
free electromagnetic field. 

\medskip 
Putting static charges into this state is accomplished by the gauge
bridges, defined in the preceding section.  
We consider here the case of two opposite charges $\pm q$ at
given points $\Xs = \{ \bxi_1, \bxi_2 \}$  
and briefly comment on more general cases further below.
Let $\beta_{q, \,  \bxi_1 \, \bxi_2}$ be 
an automorphism of $\Af(\RR^3 \backslash \Xs)$,
inducing a gauge bridge. 
We then define a state $\omega_\Xs$, containing these charges, by
composing $\omega_{\, \emptyset}$ with this gauge bridge, 
\begin{equation} \label{e4.1}
\omega_\Xs(A) \doteq \omega_{\, \emptyset} \circ \beta_{q, \,  \bxi_1 \, \bxi_2}(A) \, ,
\quad A \in \Af(\RR^3 \backslash \Xs) \, .
\end{equation}
The following properties of this state are a direct consequence 
of Lemma \ref{l3.2} and the fact that 
$\omega_\emptyset(e^{i \bEi(\bgi_\bsigma)}) = 
\omega_\emptyset(e^{i \, \divergence \, \bEi(s_\btau)}) = 1$ 
for all closed surfaces $\bsigma$. 
\begin{lemma}
Let $\omega_\Xs$ be the state defined in equation \eqref{e4.1}.  Then
\begin{itemize}
\item[(i)] $\omega_\Xs(e^{i \, \divergence \, \bEi(s)}) = 1$ for all scalar 
test functions $s$ having support in $\RR^3 \backslash \Xs$.
\item[(ii)] for any closed surface 
$\bsigma \subset \RR^3 \backslash \Xs$ and corresponding surface function
$\bgi_\bsigma$, having support in a sufficiently small neighbourhood
of $\bsigma$, one has 
\[
\omega_\Xs(e^{i \bEi(\bgi_\bsigma)}) =
\begin{cases}
1 & \! \! \! \text{if $\sigma$ encloses none or both 
points $\bxi_1, \bxi_2$} \\
e^{iq}  & \! \! \! \text{if $\sigma$ encloses
only $\bxi_1$} \\
e^{-iq}  & \! \! \! \text{if $\sigma$ encloses 
only $\bxi_2$} 
\end{cases}
\]
\item[(iii)] for any relatively compact, contractible region 
$\bOi \subset \RR^3 \backslash \Xs$, the restriction
$\omega_\Xs \upharpoonright \Af(\bOi)$ coincides with a 
vector state in the GNS representation induced by 
$\omega_{\, \emptyset}$. 
\end{itemize}
\end{lemma}
Part (ii) of this lemma shows that the state $\omega_\Xs$ describes a static
electric charge $q$ at $\bxi_1$ and $-q$ at $\bxi_2$, which both can be 
determined by flux operators (Gauss' law). According to (i), there are no other
charges present. Finally, it follows from (iii) that 
in contractible regions the state cannot be distinguished from states
in the superselection sector of $\omega_{\, \emptyset}$. Let us mention as an 
aside that these properties prevail in the causal completion 
of $\RR^3 \backslash \Xs$ in Minkowski space. 
Going to the full space would require, however, to specify 
details of the interaction between the field and the charges, which is 
not necessary here. 

\medskip
Let us turn now to the idea of Mund, Rehren and Schroer to make use of
string-localized potentials, built from observables.
In their approach such operators were concretely given as limits of
integrals involving the electromagnetic fields
$\bBi, \bEi$. We do not need such detailed information here and 
will cover arbitrary constructions of this kind, provided they are
based on observables. To this end we pick any sequence or,
more generally, net of unitary operators \ 
$\{ U_\iota \in \Af(\RR^3 \backslash \Xs) \}_{\iota \in \II}$, 
$\II$~being some index set, and consider the corresponding
automorphisms  
$$
\gamma_\iota(A) \doteq U_\iota A U_\iota^{-1} \, , \quad 
A \in \Af(\RR^3 \backslash \Xs) \, .
$$
According to standard compactness arguments, any family of states 
$\{ \omega_\emptyset \circ \gamma_\iota \}_{\iota \in \II}$ on 
$\Af(\RR^3 \backslash \Xs)$ has limit points
in the weak-*-topology, which are again states on this 
algebra. We then have the following result. 

\begin{lemma}
Let $\overline{\omega}$ be any weak-*-limit point of some family
\mbox{$\{ \omega_\emptyset \circ \gamma_\iota \}_{\iota \in \II}$}. Then
\[ 
\overline{\omega}(e^{i \bEi(\bgi_\bsigma)}) = 1
\]
for all closed surfaces $\bsigma \subset \RR^3 \backslash \Xs$ 
and all surface functions 
$\bgi_\bsigma$ which have support in a sufficiently small neighbourhood
of $\bsigma$.
\end{lemma}
\begin{proof}
According to Lemma \ref{l3.1} one has \ 
$\gamma_\iota(e^{i \bEi(\bgi_\bsigma)}) = 
U_\iota e^{i \bEi(\bgi_\bsigma)} U_\iota^{-1}  = e^{i \bEi(\bgi_\bsigma)}$, 
s$\iota \in \II$. Thus 
\[
\overline{\omega}(e^{i \bEi(\bgi_\bsigma)}) = \lim_\iota \, 
\omega_\emptyset \circ \gamma_\iota(e^{i \bEi(\bgi_\bsigma)}) = 
\omega_\emptyset(e^{i \bEi(\bgi_\bsigma)}) \, .
\] 
But $e^{i \bEi(\bgi_\bsigma)} = e^{-i \, \divergence \bEi(s_\btau)}$, so the
statement follows from the fact that the charge density
vanishes everywhere in $\omega_\emptyset$. 
\end{proof}
This lemma shows that there is no way of constructing gauge bridges
between charged fields by relying on observable
potentials. In particular, it conflicts 
with the hope that such potentials could be a
substitute for the local gauge fields in quantum electrodynamics. In
other words, local gauge fields are not 
only the solutions of cohomological problems, making
closed two-forms of observables exact. But, more importantly, 
they embody fundamental informations which are  
indispensible in physical applications. 

\medskip
We conclude this section by briefly commenting on the structure of
states containing several static charges at points
$\Xs = \{ \bxi_1, \dots , \bxi_n \}$. Assuming that the values
of these charges are $\pm q$, one connects the points 
$\bxi_k, \bxi_l \in \Xs$ by gauge bridges 
$\beta_{q_k , \bxi_k \, \bxi_l}$ with variable charges $q_k \in \RR$.
The sum of the charges $q_k$ for all paths emanating from a 
given point $\bxi_k$ must be equal to the charge $\pm q$ sitting 
there. So gauge bridges can connect a given 
charge with several other charges. One then applies the 
resulting product of gauge bridges to the state $\omega_\emptyset$
in order to obtain a state $\omega_\Xs$ for the given configuration
$\Xs$. If the total charge of
the configuration $\Xs$ is different from zero, one must
add to it charges at distant points so that 
the resulting collection $\Xs_0 \supset \Xs$ is neutral.
In case of the additional gauge bridges, appearing in 
this manner, one proceeds to limits, where the added points are
shifted to infinity. These limits 
exist because of the spatial localisation properties of the observables. 
One thereby arrives at consistent models of an arbitrary number 
of static electric charges in $\RR^3 \backslash \Xs$ 
for which Gauss' law holds. 

\section{Conclusions}
\label{sec5}

Intrigued by ideas of  Mund, Rehren and Schroer \cite{Mu,MuReSch,Sch},
we have studied in the present article the question whether local 
gauge fields can be replaced by string-localized potentials, obtained from  
local observables. Considering the case of the electromagnetic 
field in presence of infinitely heavy, 
electrically charged fields, we have shown that no kind of limiting 
procedure, involving observable potentials, can approximate the 
gauge bridges between them. These gauge bridges are 
an inevitable ingredient of the theory. In particular, they 
ensure the validity of Gauss' law, \ie the fact that 
the values of the electric charge can be determined by 
flux measurements. Thus the gauge bridges describe   
additional degrees of fredom, sometimes heuristically 
called ``longitudinal photons''. They have to be added to the 
observable string-localized potentials in one way or another. 

\medskip
Our results highlight the fact that the role of gauge fields 
reaches far beyond the cohomological problem 
of solving the homogeneous Maxwell equations. This insight 
prevails in most publications, where it is taken
for granted. Even if one is only interested in computations
of gauge invariant operators, the appearance of local gauge 
fields seems inevitable. 
For example, the construction of the electric
current by Brandt \cite{Br} in renormalized perturbation
theory, relying on point splitting techniques,
involves powers of the gauge fields, compensating 
gauge degrees of freedom of the charged matter fields in
the limit of coinciding points. 
Similarly, Steinmann in his perturbative construction of cone-localized 
charged physical fields \cite{St} had to rely on 
formal power series in terms of local gauge fields. So, together with
the present results, we come to the conclusion that the observable 
string-localized vector potentials, by themselves, are not sufficient
to describe all physically relevant features of local gauge fields. 

\section*{Acknowledgement}

\vspace*{-2mm}
DB  gratefully acknowledges the hospitality and support 
extended to him by Roberto Longo and the University of Rome 
``Tor Vergata'', which made this collaboration possible. 
FC and GR are supported by the ERC Advanced Grant 
669240 QUEST ``Quantum Algebraic Structures and Models". 
EV is supported in part by OPAL ``Consolidate the Foundations''. 
All authors acknowledge support by the MIUR Excellence 
Department Project, awarded to the Department of Mathematics, 
University of Rome Tor Vergata, CUP E83C18000100006.

\end{document}